\documentclass[reqno,11pt]{article}

\usepackage[english]{babel}
\usepackage{bbm} 
\usepackage[latin1]{inputenc}
\usepackage{amsmath,amsthm,amsfonts,amssymb}
\numberwithin{equation}{section}
\usepackage[final]{epsfig}
\usepackage{color}

\usepackage{mathtools}
\DeclarePairedDelimiter{\floor}{\lfloor}{\rfloor}

\usepackage{hyperref}
\usepackage{empheq}

\usepackage{txfonts}

\usepackage{hyperref}

\usepackage{mathrsfs}

\usepackage{enumerate}

\def\titlerunning#1{\gdef\titrun{#1}}
\makeatletter
\def\author#1{\gdef\autrun{\def\and{\unskip, }#1}\gdef\@author{#1}}
\def\address#1{{\def\and{\\\hspace*{18pt}}\renewcommand{\thefootnote}{}%
\footnote {#1}}%
\markboth{\autrun}{\titrun}}
\makeatother

\newtheorem{theorem}{Theorem}[section]

\newtheorem{lemma}[theorem]{Lemma}

\newtheorem{proposition}[theorem]{Proposition}

\theoremstyle{definition}
\newtheorem{definition}[theorem]{Definition}

\numberwithin{equation}{section}

\usepackage{graphicx,pgf}

\makeatletter
\let\Im\undefined
\let\Re\undefined
\DeclareMathOperator{\Im}{Im \,}
\DeclareMathOperator{\Re}{Re \,}

\DeclareMathOperator{\E}{\mathbb{E}}
\DeclareMathOperator{\dist}{dist}

\newcommand{\N} {{\mathbb N}}
\newcommand{\C}{\mathbb{C}}
 
\newcommand{\G}{\mathbb{G}}
\newcommand{\R}{\mathbb{R}}
\newcommand{\Schr}{Schr\"odinger }

\def\Ev#1{{\mathbb E}\left(#1 \right)}
\def\Pr {\mathbb P}    

\usepackage{bbm}

\DeclareMathOperator{\indfct}{\mathbbm{1}}

 \def\XXint#1#2#3{{\setbox0=\hbox{$#1{#2#3}{\int}$}
 \vcenter{\hbox{$#2#3$}}\kern-.5\wd0}}

\def\be{\begin{equation}}
\def\ee{\end{equation}} 

\def\={\  = \  }   
\def\+{\  + \  }   
\def\-{\  - \  }    
\def\:{\  : \   }    

\makeatother

\begin{document}




\titlerunning{Resonant Delocalization}  


\title {Boosted Simon-Wolff Spectral Criterion \\ 
and Resonant Delocalization} 

\author{Michael Aizenman  \and  Simone Warzel
}

\date{December 10, 2014 \\  Revised Dec. 7, 2015}

\maketitle

\address{M. Aizenman: Departments of Physics and Mathematics,  Princeton University, USA.
\and S. Warzel: Zentrum Mathematik, TU M\"unchen, 
 Boltzmannstr. 3, 85747 Garching, Germany.}



\begin{abstract} 
Discussed here are criteria for the existence of continuous components in the spectra of  operators with random potential.  First,   the essential condition for the 
Simon-Wolff criterion is shown to be measurable at infinity.  By implication, for the iid case and more generally potentials with the K-property the criterion is boosted by a zero-one law.   The boosted criterion, combined with 
tunneling estimates, is then applied for sufficiency conditions for the presence of continuous spectrum for random \Schr operators.  The general proof strategy which this yields is modeled on the resonant delocalization arguments  by which continuous spectrum in the presence of disorder was previously established for random operators on tree graphs. 
In another application of the Simon-Wolff rank-one analysis we 
prove the almost sure simplicity of the pure point spectrum for  operators with random potentials of conditionally continuous distribution.
\end{abstract}


\section{Introduction} 

Studies of the spectral effects of  disorder often deal with self-adjoint operators of the form
\be  \label{H}
H(\omega) = A + V(\omega) 
\ee 
acting in the $\ell^2$-space of functions over an infinite graph $\G$, where  
$ A $ is a self-adjoint bounded   operator and  $V(\omega) $ is a multiplication operator by a random function (the random potential)
$
[V(\omega) \psi ](x) \ = \ V(x;\omega) \psi (x)
$
with $\{ V(x;\omega)\}_{x\in \G}$ a collection of random variables with a specified joint distribution.   The parameter $\omega$, which represents the disorder, 
ranges  over a standard probability space $(\Omega, \mathcal B, \mathbb{P} )$.  In this setup $H(\omega)$ forms a weakly measurable, self-adjoint, operator-valued function  (cf.~\cite{CaLa}).  The strength of the disorder is expressed here in the width of the distribution of $V(x;\omega)$, loosely speaking.

The spectral measure associated with a specified realization of 
$H(\omega)$ and a vector
$\psi \in \ell^2(\G)$,  
is
defined by the property: 
\be 
 \langle \psi,\,  F(H(\omega))\,  \psi \rangle \ = \  \int _\R F(E) \, \mu_{\psi}(dE;\omega)
\ee 
for all $F \in C_0(\R)$ (continuous  function which vanishes at infinity).  
Each  measure can be decomposed into 
its  pure-point component and a continuous one: 
\be 
 \mu_{\psi} \ = \  \mu_{\psi}^{pp} \ +\  \mu_{\psi}^{c}
 \ee 
where $ \mu^{pp}$ is a countable sum of  point measures and $\mu^{c}$ is a continuous remainder.   
The distinction between the two  spectra is reflected in the nature of the (possibly generalized) eigenfunctions: in the pure point case these are proper elements of the $\ell^2$-space, whereas for the continuous spectrum the eigenfunctions are not  square summable.  
The difference carries also significant implications for the recurrence properties of the unitary evolution generated by $H(\omega)$ (cf.\ \cite{CFKS}) and  the conductive properties of particle systems with such one-particle Hamiltonians.   \\

In the well known Anderson localization phenomenon~\cite{an},   at sufficiently high disorder 
 as well as at extremal energies (with some exceptions~\cite{AiWa_EPL}), the spectrum of $ H(\omega) $ is almost surely of pure point type, consisting  of dense (random)  collections  of proper eigenvalues  associated with  square integrable eigenfunctions.  
There remains however  dearth of methods for establishing regimes of delocalization in the presence of disorder.   
On the short list of  such  are  arguments based on  \emph{resonant delocalization}.  
   This approach has been especially effective for random 
Schr\"odinger operators on tree graphs~\cite{AiWa_EPL,AiWa_resdeloc}, but it was shown to guide one to correct conclusions also in other contexts~\cite{ASW_14}. 
Our main goal here is to advance this method, combining it  with an improved version of 
the Simon-Wolff criterion for a related sufficiency criterion under which one may conclude  
 the existence of continuous spectrum, and in some situations absolutely continuous one.
  \\  
 
In a related application of the Simon-Wolf criterion for point spectrum, in Appendix II we present an improved result on the simplicity of the point spectrum proving it for a naturally broad class of random potentials.  

\section{The Simon-Wolff spectral criterion and its boost}  
\label{sec:SiWo} 

\subsection{The Simon-Wolff  sufficiency condition for continuous spectrum}

Analysis of the spectral measures associated with the canonical basis elements $\delta_x\in \ell^2(\mathbb{G})$ 
is facilitated by considerations of the  Green function: 
\be  \notag
G(x,y;z;\omega) := \langle \delta_x \, , (H(\omega) - z )^{-1} \delta_y \rangle \,.  
\ee 
When  the potential of $H(\omega)$  is changed at a site $x\in G$   by $\delta V(x)$, 
an  eigenfunction which solves 
\be (H-E) \varphi  (y)=0 
\notag
\ee  
turns into a solution of  the Green function equation 
\be \notag 
([H +\delta V(x) P_x]  -E) \varphi  (y)\ = \ [\delta V(x) \, \varphi(x)] \delta_{x,y} \, . 
\ee 
This elementary observation underlines a number of results concerning the structural similarity of the eigenfunctions to the kernel of the Green function, with one of its arguments fixed.   
In particular, starting from Aronszajn's analysis of rank-one perturbations~\cite{Aron57}, B. Simon and T. Wolff  noted the following~\cite{SimWol}.
 
 \begin{proposition} \label{prop:Aronj}
Let $\G$ be a countable  set,  $H_0$ a bounded self-adjoint operator in $\ell^2(\G)$, and $H_v$ the one-parameter family of operators defined by  
\be  \label{eq:H_v}
 H_v = H_0  + v \, P_\psi 
\ee 
with $P_\psi$ a rank-one orthonormal projection on a space spanned by a normalized function $\psi$.    
Then for any $ v \neq 0 $ and $ E \in \R $ the following statements are equivalent:
\begin{enumerate}[i.]
\item $E$ is a proper eigenvalue of $H_v$, i.e. $ \mu_{v,\psi}\left(\{ E\} \right) > 0 $,\\[-0.3ex]  
\item the following quantity is finite
\be \notag 
 \gamma_{0,\psi}(E) \ := \ \lim_{\eta \downarrow 0} \sum_y | \langle \delta_y , \, (H_0 - E -i \eta ) ^{-1} \, \psi \rangle |^2  \ = \    \int \frac{\mu_{0,\psi}(dt;\omega)}{(t-E)^2 }       \ < \  \infty  
\ee  
and 
\be 
\langle \psi , \, (H_0 - E -i 0 ) ^{-1} \, \psi \rangle \ = \ - v^{-1} \, .
\ee 
Moreover, if the condition is satisfied  then  $ \mu_{v,\psi}\left(\{ E\} \right) = v^{-2}  / \widehat \gamma_{\psi}(E) $.  
\end{enumerate}
\end{proposition}

Combining  the above with a spectral averaging principle, Simon and Wolff  
presented a useful criterion in which reference is made to 
\be \label{def_gamma}
 \gamma_{x}(E;\omega) \ := \ \lim_{\eta \downarrow 0} \sum_y |G(x,y;E+i\eta;\omega)|^2   = 
\int \frac{\mu_{\delta_x}(dt;\omega)}{(t-E)^2}  \, \in \ (0, \infty]\,  
\ee 
(the limit existing by monotonicity).

\begin{proposition}[Simon - Wolff \cite{SimWol,SimR1}] 
\label{thm:SimWol}  
Let $H(\omega)= A +  V(\omega) $ be a self-adjoint operator on $\ell^2(\G)$ such that for each $x\in \G$ the random variable $ V(x) $ is of conditionally absolutely continuous distribution,  conditioned on   ($V_{\neq x}:= \{ V(y)\}_{y\neq x}$). 
If, for a Borel subset   $I \subset \R$ 
\be \label{eq:SimWolcond2}
 \gamma_{x}(E,\omega)  \ = \  \infty  \,  ,
\ee 
for Lebesgue almost every $E \in I$  and $ \mathbb{P} $-almost all $ \omega $,  
then almost surely  $H(\omega)$ has only continuous spectrum  (if any) in $I$.
\end{proposition} 

For the proof, and hence also applications of this criterion it is essential  that the 
event~\eqref{eq:SimWolcond2} is of probability one.  
Yet  the second-moment analysis which has been employed in resonant delocalization arguments yields (initially) only a weaker result, that \eqref{eq:SimWolcond2} holds with non-zero probability.   
The purpose of the following boost is to bridge this gap. 
 
\subsection{The boost: a zero-one law}
 
 Let $V(\omega)$ be a random potential, specified through the collection of random variables $\{V(x;\omega)\}_{x\in \G}$.  For each $\Lambda \subset \G$, we denote by $\mathcal B_{\Lambda}$ the minimal $\sigma$-algebra of  subsets $A\subset \Omega$ for which  $\omega \mapsto 1_A(\omega)$ is a measurable function of  $\{ V(y)\}_{y\in  \Lambda}$.

\begin{definition}  \mbox{} \\[-3ex] 
 \begin{enumerate}[1.]  
\item A random variable $F: \Omega \mapsto \R$ is \emph{measurable at infinity} if    for each finite $\Lambda \subset \G$,  $F$  is measurable with respect to $\mathcal B_{\Lambda^c}$.   
\item A stochastic process over a graph $\G$   is said to have  the \emph{$K$-property}, if any random variable which is measurable at infinity is constant almost surely. 
\end{enumerate}
\end{definition} 

The simplest example of processes with the $K$-property are those for which $\{V(y)\}_{y\in \mathbb{G}} $ are independent random variables.  For random potentials with this property the applicability of the Simon-Wolff criteria 
is hereby  boosted by the following zero-one law.  

\begin{theorem}  \label{thm:0-1}
Let $H(\omega)= A +  V(\omega) $ be a random self-adjoint operator on $\ell^2(\G)$ with  $V(\omega)$ a random potential with the $K$-property, and such that for each vertex $x\in \G$  the conditional single-site distribution of $ V(x) $ conditioned on 
$V_{\neq x} $ is  continuous.
Then  for Lebesgue almost every $E\in \R$: 
\be \label{eq:0-1}
 \mathbb{P} \left(  \gamma_{x}(E) < \  \infty  \, \right)    \quad \mbox{equals either $0$ or $1$.  } 
\ee 
\end{theorem} 

As will be explained in the proof,   the condition 
$\{ \gamma_{x}(E) < \  \infty \} $ 
is essentially equivalent to: 
\be  \label{kappa}  
\kappa_x(E,\omega)  \ := \  \lim_{\eta \downarrow 0}  \frac{-1}{\eta} \Im \,  \langle \delta_x , \frac{1}{H(\omega)-E-i \eta} \,  \delta_x\rangle^{-1} \ <  \  \infty    \,  .  
\ee 
In effect, the proof of \eqref{eq:0-1} proceeds by  showing that for fixed $E\in \R$  the set 
$ \{ \omega \in \Omega \, \big | \,  \kappa_x(E,\omega))   < \infty \}  $ is measurable at infinity -- in the Lebesgue sense, that is up to corrections by sets of measure zero.   

The rest of this section is devoted to  the details of the  argument outlined above.   
Other than the conclusion summarized in Theorem~\ref{thm:0-1}, the proof does not play a role in our discussion of  resonant delocalization, which is the subject of Section~\ref{sec:deloc_tunneling}.

\subsection{Measurability at infinity}

In the proof of Theorem~\ref{thm:0-1} we shall make use of  rank-one and rank-two perturbation formulae, which express the  dependence of $G(x,x, z;\omega)$ on $V(x)$, and its joint dependence on $V(x)$ and any other $V(y)$: 
\begin{enumerate}[1.]
\item The dependence on the potential at $x$ is particularly simple:
\be \label{rank_1}
\langle \delta_x , (H-z)^{-1}\,  \delta_x\rangle \ =: \ \left[ V(x) - \Sigma(x; z)\right]^{-1}
\ee 
with $\Sigma(x; z)$, which is referred to as the \emph{self-energy}, a function of $V_{\neq x}$ only.    
\item For any pair of distinct sites, $x\neq y$:
\begin{align}\label{eq:rank2}
 \left( \begin{matrix} G(x,x;z) & G(x,y;z) \\ G(y,x;z) & G(y,y;z) \end{matrix} \right) =: \left(  \begin{matrix} V(x) -  \sigma(x;z) & - \tau(x,y;z) &  \\ - \tau(y,x;z) & V(y) - \sigma(y;z) \end{matrix}\right)^{-1} \, .
\end{align}  
with $\sigma(x;z) $ and $\tau(x,y;z)$ which do not involve $V(x)$ and $V(y)$.   
Following \cite{ASW_14} we refer to $\tau(x,y;E+i0) $ as the (pairwise) tunneling amplitude between the two sites, at energy $E$. 
\end{enumerate} 
These expression  form two special cases of  the Schur-complement, or Krein-Fesh\-bach formula.   
In the discussion of their implication on the properties of $\kappa_x(E)$ the following statement will be of relevance.   

   \begin{lemma}  
\label{lem:Fv} Let 
\be  \label{frac_lin}
F_n(V) \ := \ \frac{a_n V + b_n}{c_n V + d_n} 
\ee  
stand for a sequence of M\"obius functions with the  property that for all  $V \in \R$:
\begin{enumerate}[1.]
\item $ \Im F_n(V) \geq 0 $, and
\item $\Im F_n(V) $ converges to a limit within $[0,\infty]$ (allowing the value $+\infty$).  
\end{enumerate}
Then, $\displaystyle\lim_{n\to \infty} \Im F_n(V)  $ is finite or infinite simultaneously for all, except at most one value of $V\in \R$.    
\end{lemma}

\begin{proof} 
 The fractional linear mapping $F_n : \, \R \to \{ z \in \C | \Im z \geq 0\} =: \mathbb{C}^+_0$    takes $\R$  (or rather its one-point compactification $\dot \R = \R \cup \{ \infty\} $) onto a generalized circle (possibly a line) in $  \mathbb{C}^+_0 $, preserving the canonical orientation.   We denote the circle's radius by $R_n \in [0,\infty]$ and its lowest point by 
$U_n \in \C^+_0$.
In the  degenerate case,   $R_n=\infty$, we set $\Re U_n=0$.    
The asserted dichotomy  holds trivially true (and without exceptional points) if either 
\begin{enumerate}[1.]
\item 
$\displaystyle \limsup_{n\to \infty}  \Im U_n = \infty$ 
 (in which case $\displaystyle \limsup_{n\to \infty} \Im F_n(V) \ = \ \infty  $ for all $ V\in \R$),
  or   
\item   $(\Im U_n )$ and  $( R_n)$ are bounded sequences  
 (in which case $\displaystyle \limsup_{n\to \infty} \Im F_n(V) \ < \ \infty $ for all $ V\in \R$). 
\end{enumerate} 
  Hence it suffices to establish the claim for the case that 
$(\Im U_n)  $ is bounded  and $\displaystyle \limsup_{n\to \infty}  R_n = \infty$. 
Furthermore, since when $\displaystyle \lim_{n\to \infty} \Im F_n(V)  $  exists    it  can be computed over any subsequence,  it suffices to prove the assertion under the additional assumption that these limits exist, i.e.   
\be \notag 
 \lim_{n\to \infty} R_n = \infty \, , \qquad \lim_{n\to \infty}  \Im U_n < \infty \,. 
 \ee 
 As a final simplification we note that the assertion holds for $F_n$ if and only if it holds for the sequence of shifted functions $ F_n(V) -U_n$.
 Based on these considerations,  we add without loss of generality the assumption that  $ U_n = 0$  for all $n\in \N$.
 
 An equivalent form of the statement to be proven is:    if  
\be   \label{Vj}
\lim_{n\to \infty} \Im F_n(V) \ < \ \infty 
\ee 
for two distinct values $V_1<V_2$, then \eqref{Vj} holds for all, but at most one value of $V\in \R$.   We shall prove it in this form.

 Let $V_j$ with $j=1,2$ be  two points  for which \eqref{Vj} holds, and let $Y_j := \displaystyle \lim_{n\to \infty} \Im F_n(V_j)$.  The convergence of the imaginary part does not ensure that of the real part, which may still oscillate between the region $ \Re F_n(V) \geq 0 $ and $ \Re F_n(V) < 0 $. As explained above, it suffices to restrict attention to a subsequence, and we select one  for which the signs of the real part take consistently fixed $\pm1$ values, i.e. for all $ n $ and both $j=1,2$:
  \be \label{sigma} 
{\rm sign}\,  \Re F_n(V_j) \ = \ \sigma_j
 \ee 
declaring $\text{sign}\,  0 = +1$.  

For  the following argument it is convenient to have one more point with the properties of $V_j$, and we start by showing that such a point exists.

The removal from $\dot \R $  of $V_1$ and $V_2$ splits the `generalized circle' which is the image of $\R$ under $F_n$ into two arcs.  We shall refer to  the one which includes the point which minimizes $\Im F_n(V)$ as the `lower arc'.     Along this arc the value of $\Im F_n $  is everywhere bounded by $\max \{ \Im F_n(V_1) , \Im F_n(V_2) \}$.  
There are now  two possibilities: the lower arc will contain  either the image of the midpoint $  (V_1 + V_2)/2 $, or else the image of $V=\infty$, in which case it will also contain the image of $ V_1 - (V_2-V_1)/2 $ (i.e., the midpoint's reflection about $V_1$).   
Restricting to a subsequence, we may assume without loss of generality that which of the options applies does not change with $ n $, and we shall pick the point $ V_3\in \R $ as either  $ (V_1 + V_2)/2 $ or  $ V_1 - (V_2-V_1)/2 $, corresponding to this alternative. 
By this construction, $ \Im F_n(V_3) \le  \max \{ \Im F_n(V_1) , \Im F_n(V_2) \}$,  so that 
\eqref{Vj} holds also for $j=3$.   Possibly passing to a subsequence, it may  be assumed that for one of the  choices of $\sigma_3 \in \{-1,1\}$  also  \eqref{sigma} holds consistently  for $j=3$.

By the invariance of the cross ratio under  fractional linear transformations, for any $V\in \R$ and the above $V_j$: 
 \be  \label{cross_ratio}
 g(V; V_1,V_2,V_3) \ := \  \frac{(V-V_1)(V_2-V_3)}{(V-V_2)(V_1-V_3)} \ = 
 \ \frac{[F_n(V)-F_n(V_1)] [F_n(V_2)-F_n(V_3)]} {[F_n(V)-F_n(V_2)] [F_n(V_1)-F_n(V_3)]} \,.  
 \ee  
We now will show that under the above assumptions, keeping the above three points fixed as $n\to \infty$:
\be \label{ratio1}
\lim_{n\to \infty} \, \frac{F_n(V_2)-F_n(V_3)} {F_n(V_1)-F_n(V_3)} \ = \   
\frac{\sigma_2 \sqrt{Y_2}-\sigma_3\sqrt{Y_3}} {\sigma_1 \sqrt{Y_1}-\sigma_3\sqrt{Y_3}} 
\ee  
whereas for all $V\in \R$ with $\Im F_n(V) \ge T \ge \ 2  W$, with $W :=\max \{ Y_1, Y_2\}$,
\be \label{ratio2}
 \limsup_{n\to\infty} \left|   \frac{F_n(V) -F_n(V_1)} {F_n(V) -F_n(V_2)} \ - \   1 \right| \ \le  \ C \sqrt{W/T} \\   
\ee 
with a fixed constant.

 The   relation \eqref{ratio1} is based on the observation that the coordinates of $F_n(V_j) =: X_{n,j} +i Y_{n,j}$ satisfy:  
 \be  \notag  
 X_{n,j}^2 + (R_{n}-Y_{n,j})^2 = R_{n}^2 \, ,  \quad \mbox{and}  \quad \text{sign}\,  X_{n,j} = \sigma_j
 \ee 
 and hence 
 \be  \notag 
 X_{n,j} = \sigma_j \, \sqrt{2R_n Y_{n,j}} \, \sqrt{1 - Y_{n,j}/(2 R_n)} \, . 
 \ee 
 Since $R_n \to \infty$, equation \eqref{ratio1} easily follows.  In essence, these estimates reflect the flatness, and `horizontality' of the curve near its bottom, due to the asymptotical vanishing  of the circle's curvature.  Related considerations yield \eqref{ratio2}.    
 
The relations \eqref{cross_ratio}, \eqref{ratio1}, and \eqref{ratio2}, imply that for each fixed $T > 2 W$, all $V\in \R$ with $\Im F_n(V) \ge T$ at $n$ large enough: 
\be \label{rbound}
 \left| g(V; V_1,V_2,V_3)  -\frac{\sigma_2 \sqrt{Y_2}-\sigma_3\sqrt{Y_3}} {\sigma_1 \sqrt{Y_1}-\sigma_3\sqrt{Y_3}}   \right|  \ \le \  C  \sqrt{W/T}\, \, \,  
 \frac{\sigma_2 \sqrt{Y_2}-\sigma_3\sqrt{Y_3}} {\sigma_1 \sqrt{Y_1}-\sigma_3\sqrt{Y_3}} 
\ee 
The cross ratio of $V$ with fixed $\{V_j\}_{j=1,2,3}$ forms a continuous  $\dot \R $ valued function over $\dot \R$ which  determines V uniquely.  The intersection of the bound \eqref{rbound} over a  sequence of values of $T\to \infty$ implies that the set of points for which $\displaystyle\limsup_{n\to \infty} \Im F_n(V) = \infty$ consist of only the one point for which 
\be  \notag  
 g(V; V_1,V_2,V_3)  = \frac{\sigma_2 \sqrt{Y_2}-\sigma_3\sqrt{Y_3}} {\sigma_1 \sqrt{Y_1}-\sigma_3\sqrt{Y_3}}   \, .  
\ee
  \end{proof} 

\mbox{} \\[-3ex]  

Using Lemma~\ref{lem:Fv} we now turn to prove the zero-one law which was introduced  above. 

\begin{proof}[Proof of Theorem~\ref{thm:0-1}] 
For a convenient  reformulation of the condition in \eqref{eq:0-1}, let 
\be \label{def:F}
\mathcal {F}_x(z,\omega)  \ := \ \frac{-1}{ \Im z}  \, \langle\delta_x , \frac{1}{H(\omega) -z} \, \delta_x\rangle^{-1} \,  .
\ee
It satisfies: 
\be \label{F}
\Im \mathcal {F}_x(E+i\eta,\omega)   \ = \ 
\left |G(x,x;E+i\eta;\omega)  \right|^{-2}  \ 
\langle\delta_x , \frac{1}{[H(\omega) -E]^2+\eta^2} \, \delta_x\rangle
\ee  
For a full measure of energies $ E \in \R $, the limit $G(x,x;E+i0;\omega)  := \lim_{\eta \downarrow 0}   G(x,x;E+i\eta;\omega) $ exists and is finite non-zero for $\mathbb P$-almost every $\omega$.  (By Fubini's theorem the statement is equivalent to its reversed-order form, and that is implied by the  de la Vall\'ee-Poussin theorem.)   It follows that for $E $ in this full measure set: 
\be  \label {gamma_kappa} 
 \mathbb{P} \left(  \gamma_{x}(E) < \  \infty  \, \right)   \ = \   \mathbb{P} \left(  \kappa_{x}(E) < \  \infty  \, \right)  
\ee 
(for the quantities defined in \eqref{def_gamma}  and \eqref{kappa}).  
We now proceed restricting our attention to this regular set of energies $E$. 

Let us denote the event: 
\be  \notag  
 \mathcal K_x(E) \  := \  \left\{\omega \in \Omega \; \big| \,   
 \kappa_x(E;\omega) \ <  \  \infty   \right\} \,  . 
\ee 
Using Lemma~\ref{lem:Fv} we shall now prove that 
for every finite set $ \Lambda \subset \G$ with $\{x\} \subset \Lambda $, the regular conditional probability  of $\mathcal K_x $ conditioned on $\mathcal B_{\Lambda ^c} $ is either zero or one, or more explicitly:
\be \label{claim} 
\  \E\left[ \, 1_{\mathcal K_x}   \, |\,  \mathcal B_{ \Lambda ^c} \,\right] (\omega) 
 \ \stackrel{a.s.} {=} \  1_{\mathcal K_x}(\omega)  \, ,     
\ee 
where the quantity on the left is the conditional probability 
$\mathbb{P} { \left ( \, \mathcal K_x \, |\, \mathcal B_{\Lambda ^c} \, \right ) }(\omega) $, expressed 
 as an average in which  the variables $\{V(u)\}_{u\in \Lambda}$ are integrated out with their appropriate conditional distribution.   
  
 Lemma~\ref{lem:Fv} is applicable here due to \eqref{eq:rank2}.  It tells us that for any site $u \in \G$ the indicator  function $1_{\mathcal K_x}$ is almost surely constant as $V(u) $  is varied at fixed values of $V_{\neq u}$.  Due to the continuity of the conditional distribution, the set of $\omega$ for which $V(u;\omega) $ takes one of the exceptional values (of which there are at most two) is of zero probability.
  It follows that for any $u\in \G$:
\be  \label {induction_1}
 \E\left[ \, 1_{\mathcal K_x}   \, |\,  \mathcal B_{ \{u\} ^c} \,\right] (\omega)  \ 
 \stackrel{a.s.} {=}     \   1_{\mathcal K_x} (\omega)  \, .  
\ee  

Denoting by $\mathcal P_\Lambda $  the orthogonal projections in $L^2(\Omega,  \mathbb P)$  corresponding to the conditional expectations:  $\psi \mapsto  \E\left[ \, \psi    \, |\,  \mathcal B_{ \Lambda ^c} \,\right] $, Eq.~\eqref{induction_1} can  be  equivalently stated in the form: 
\be  \label {induction_1b}
 \E\left[ \,|  1_{\mathcal K_x}  - \mathcal P_{\{u\}}  1_{\mathcal K_x} |^2   \,\right]   \  = \ 0 .  
\ee  
By an elementary orthogonality bound, for any finite (and by implication also infinite) $\Lambda \subset \G$: 
\be  \notag  
 \E\left[ |\, 1_{\mathcal K_x}  - \mathcal P_\Lambda 1_{\mathcal K_x} |^2 \right] \ \le 
     \ \sum_{u\in \Lambda}   \E\left[ |\, 1_{\mathcal K_x}  - \mathcal P_{\{u\} } 1_{\mathcal K_x} |^2 \right] \, .
\ee 
 Therefore, \eqref{induction_1b} implies that the above quantity vanishes also for all finite $\Lambda$.   This proves \eqref{claim} for finite $\Lambda \subset \G$.   Through the martingale convergence theorem (or  through just an extension of   the above variance bounds)  we conclude  that  for any $ \Lambda \subset \G $:  
 \be  \label {cond_at_infty}
 \E\left[ \, 1_{\mathcal K_x}   \, |\,  \mathcal B_{\Lambda^c} \,\right] (\omega)  \ 
 \stackrel{a.s.} {=}     \   1_{\mathcal K_x} (\omega)  \, ,   
\ee  
and thus  $1_{\mathcal K_x} $ is measurable at infinity.  Under the assumption that the joint distribution of the potential has the $K$-property it follows that $ \mathbb{P} \left(  \kappa_{x}(E) < \  \infty  \, \right)  $ can equal only $0$ or $1$, and through  \eqref{gamma_kappa} this implies  
the claim \eqref{eq:0-1}. 
\end{proof} 

\subsection{A weak form of spectral dichotomy}  

Theorem~\ref{thm:0-1} along with with the Simon-Wolff criterion, Theorem~\ref{thm:SimWol},  imply: 
\begin{enumerate}[1.]
\item The real line is covered up to a zero measure subset by the disjoint union $ \mathcal{C}_x \cup \mathcal{P}_x $ of the non-random sets: 
\begin{align}\label{eq:defCx}
\mathcal{C}_x & := \left\{ E \in \R \, | \, \mathbb{P}\left(\gamma_x(E)=\infty\right) = 1 \right\}  \notag \\
\mathcal{P}_x & := \left\{ E \in \R \, | \, \mathbb{P}\left(\gamma_x(E)<\infty\right) = 1 \right\} 
\end{align}

\item With probability one $ \mathcal{C}_x $ serves as a support for the continuous spectrum of $ H(\omega) $ in the sense that:
\begin{enumerate}[i.]
\item $ \mu^{pp}_{\delta_{x}}( \mathcal{C}_x;\omega) \ = \  0 $ for $ \mathbb{P} $-almost all $ \omega $,
\item for any $\varepsilon >0$ and Lebsgue almost every  $E\in \mathcal C_x$ :  \be \label{cond_Q2}
\mu^c_{\delta_x}((E-\varepsilon, E+\varepsilon) ;\omega) \ > \  0  \, \quad \mbox{for $ \mathbb{P} $-almost every $ \omega $}. 
\ee   
\end{enumerate}
\item $ \mathcal{P}_x $ supports the pure-point spectrum of $ H(\omega) $ together with the real part of the resolvent set. In particular, 
\be  \mu^{c}_{\delta_{x}}( \mathcal{P}_x;\omega) \ = \  0 \qquad \mbox{ for $ \mathbb{P} $-almost all $ \omega $.} 
\ee  
\end{enumerate}

One may add to it that  the condition  $ \lim_{\eta\downarrow 0} \langle \delta_x , \frac{1}{H(\omega) -E-i \eta} \,  \delta_x\rangle  >   0 $,  in which existence of the limit is part of the statement, is also measurable at infinity -- a fact which has already been noted before (\cite[Cor.~1.1.3]{JLInv}).     
Denoting
\begin{align} \notag 
\mathcal{A}_x \  := \  \left\{ E \in \R \, | \, \mathbb{P}\left( 
 \Im \,  \langle \delta_x , \frac{1}{H-E-i 0} \,  \delta_x\rangle  >   0  \,   \right) = 1 \right\}  \, , 
\end{align} 
and observing that 
$  \mathcal{A}_x  \subset   \mathcal{C}_x$ (which follows from the spectral representation),  one may add to the above: 
\begin{enumerate}[1.]
\item[4.]  the support of the continuous spectrum admits also a non-random disjoint decomposition, with : 
\begin{enumerate} [i.]
\item 
$\mathcal{A}_x$ providing  
an  almost sure support of the absolutely continuous spectrum 
\item $\mathcal{C}_x  \backslash   \mathcal{A}_x $ serving as an   almost sure support of the singular continuous spectrum.  
\end{enumerate} 
\end{enumerate} 
For the last point (4)  we recall that the spectral measure's absolutely  continuous component  is 
\be \notag 
\mu^{ac}_{\delta_x}( dE) \ = \  \pi^{-1}  \Im \,  \langle \delta_x , \frac{1}{H-E-i 0} \,  \delta_x\rangle  \, dE \, . 
\ee 

Thus, the boosted Simon-Wolff criterion  implies a measure theoretic  form of spectral dichotomy.   It should however be appreciated that  $\mathcal{C}_x  $ and $\mathcal{P}_x $ generally do not coincide with the topological definition of the continuous and pure point spectra since these sets may in general  not be closed.   In this context, let us recall that the non-randomness of the topological supports  of the different spectra is also known quite generally for ergodic operators~\cite{CFKS,CaLa,PF}.

 \section{Resonant delocalization}  \label{sec:deloc_tunneling}  

Delocalization in the presence of extensive disorder turned  to be more elusive than Anderson localization.  Challenges remain open at both the level of compelling physical argument and of mathematical existence proofs of spectral regimes with extended states for random Schr\"odinger operators in finite dimensions.   
In particular, the  Bloch-Floquet mechanism for the formation of bands of ac spectrum for periodic operators  
is unstable with respect to even weak homogeneous disorder. That is drastically demonstrated by the one dimensional example, where at arbitrarily weak disorder the entire spectrum changes  to pure point~\cite{gmp,CFKS}.

 In the following, we build on one of the very few effective arguments for the formation of continuous spectrum through local resonances to have emerged in  mathematical works on delocalization.  
For simplicity, we focus on operators of the form \eqref{H} with independent and identically distributed (iid) potentials, whose single-site distribution is 
absolutely continuous with bounded density $ \rho $, and  assume homogeneity in the following sense.

\begin{definition}   A self-adjoint random operator  on a transitive graph $\G$
is said to be \emph{homogenous} if for each $T$ in a transitive collection of graph homomorphisms of $\G$   
the action of  $T$ on $\G$ lifts to measure-preserving transformation on $\Omega$   
under which $H(T \omega) $ is unitarily equivalent to $H(\omega)$ and satisfies: 
\be\label{eq:translcovariant}
\langle \delta_{T(x)} , F(H(\omega)) \delta_{{T(x)}} \rangle \ \stackrel{a.s.}{=}\  \langle \delta_{x}, F(H(T\omega)) \, \delta_{x} \rangle \, ,
\ee
for all bounded continuous $ F: \R \to \R $ and all $ x \in \G $.  
\end{definition} 

In this set-up, the mean density of states measure, which is generally defined  by 
$ \nu(I ) := \mathbb{E}\left[ \mu_{\delta_x}(I) \rangle \right] $, for Borel $ I \subset \R $, does not depend on $ x \in \G $, 
and is known to be absolutely continuous, with a bounded derivative  satisfying
\be  \notag 
n(E)  := \frac{\nu(dE)}{dE}  \ \le \ \| \rho\|_\infty \, .  
\ee
(The function $n(E)$ is referred to as the density of states;  the estimate is known as the Wegner bound, and its proof can be based on the spectral averaging principle \cite{SimWol,SimR1}).

By the  zero-one law of Theorem~\ref{thm:0-1} and the Simon-Wolff criterion, in order to establish the presence  of continuous spectrum within an interval $I$ it suffices to prove that for a positive measure set of energies $E\in I$: 
\be \label{ell2_div}
 \mathbb{P}\left\{ \lim_{\eta \downarrow 0} \sum_x |G(0,x;E+i\eta;\omega)|^2    = \infty \right\}  \ >\ 0   \, .  
 \ee 
 where the sum can also be written as $\langle\delta_0 , (H(\omega) -E]^2+\eta^2)^{-1} \, \delta_0\rangle$.  
The representation provided in \eqref{F} makes it clear that, with the  possible exception of a zero measure set of energies, the above limit diverges for each $E$ in the set
\be  \notag 
\mathcal{A}_0 = \left\{ E \in \R \, | \, \mathbb{P}\left( \Im  G(0,0;E+i0) \  > \ 0 \right) \neq  0 \right\} \,.
\ee 
Our discussion will therefore focus on conditions implying the divergence under the  assumption that $G(0,0;E+i0,\omega) $ is real for $ \mathbb{P} $-almost all $ \omega $.   (It is not difficult to see that  this event also satisfies a zero-one law for almost every $ E \in \R $, but this observation will not be needed here.)    \\  

For brevity of notation, from this point on we shall omit the explicit reference to the limit, denoting: 
\be \notag 
G(x,y;E)  \ := \ \lim_{\eta \downarrow 0} G(x,y;E+i\eta) \,. 
\ee 
The limit  exists for almost all $(E, \omega)$  simultaneously  for all $ x, y \in \G $ (as follows from the de la Vall\'ee-Poussin theorem).

\subsection{Rare but destabilizing resonances}

A simple lower bound for the quantity of interest is, 
for any $ R \in \mathbb{N} $: 
\begin{eqnarray}  \label{eq:lowergamma}
\gamma_0(E;\omega) &=&  \lim_{\eta \downarrow 0 }\,  \sum_{x \in \G} \left| G({0},x;E+i\eta,\omega) \right|^2   \geq \!  
\notag \\ 
 &\geq & \sum_{x: d({0},x) = R  } \left| G({0},x;E) \right|^2  
 \ =\  \left| G({0},{0};E) \right|^2 \!  \sum_{x : d({0},x)  = R  }  \left| g(x;E)\right|^2  \, .   
\end{eqnarray}
with (in terms introduced in \eqref{eq:rank2}): 
\begin{align}\label{eq:defratioresloc}
g(x;E)   := \frac{G({0},x;E)}{G({0},{0};E)} \ = \ \frac{\tau({0},x;E)}{V(x) -\sigma(x;E)} \, .  
\end{align} 

We shall now present a scenario under which a given site $ 0 $ is resonant at energy $E$ with a random collection of many other sites $ x \in \G$, for which  $ |g(x;E)|\approx 1 $.  
The resonances on which we shall focus are expressed in the 
 joint occurrences of the following three events: 
\begin{align}  \label{TEN}
& \mathcal{T}_x := \left\{\,  \left| \tau(0,x;E)\right|  \geq t(0,x;E) \, \right\} \,  \notag \\
& \mathcal{E}_x := \left\{ \, \left|V(x) - \Sigma(x;E) \right| \leq t(0,x;E) \, \right\} \,  \notag \\
& \mathcal{N}_x := \left\{ \, |V(0) - \sigma(0;E) | \geq  |\tau(x,0;E)|\, \right\} \, ,
\end{align}
with $\Sigma(x;E)$ defined by the rank-one relation \eqref{rank_1}, and $ t(x,y;E) \in(0,1]$  selected so that: 
\be \label{eq:typdecaytau}
\lim_{R\to \infty} \min_{d(x,y)=R} \mathbb{P} \left( \,  \left|\tau(x,y;E) \right| \ \geq \  t(x,y;E)  \, \right) \ =1   \, .  
\ee 
A related quantifier of the tunneling amplitude's distribution which will play a role is its truncated average   
\begin{align} \label{t_T}
T(x,y;E) \ &:=   \mathbb{E}\left[ \min\{ |   \tau(x,y;E) | , 1 \} \right]  \geq     c \,    t(x,y;E) \, ,
\end{align}
which for any $\varepsilon >0$ satisfies:   $T(x,y;E)  \ge (1-\varepsilon) \, t(x,y;E) $ at  sufficiently  large  distances, $d(x,y)\geq R_\varepsilon$. \\ 

The events $ \mathcal{T}_x $ and $ \mathcal{N}_x$   do not depend on $V(x)$, and in the situation  discussed below, they will be found to occur with asymptotically full probability as $d(x,0)\to \infty$.   In contrast,  $\mathcal{E}_x$ depends on the value of $V(x)$, and requires it to fall within an extremely narrow range of values  (near $\Sigma(x;E)$ which  depends on $V_{\neq x}$).     A key fact is that  
at any site $x\in \G $ for which the three condition are met for a given potential, i.e. under the event $  \mathcal{T}_x \cap \mathcal{E}_x  \cap  \mathcal{N}_x $, the ratio  which appears in~\eqref{eq:lowergamma} is bounded below:
\be  \label{eq:r_1_2}
| g(x, E) | \geq \frac{ 1 }{ 2 }  \,. 
\ee
For a proof let us note that \eqref{rank_1} with \eqref{eq:rank2} yield:
\be\label{eq:defSigmardel}
\Sigma(x;z) \ = \  \sigma(x;z) + \frac{ \tau(x,0;z)\, \,   \tau(0,x;z)}{V(0) - \sigma(0;z) } \, , 
\ee
which shows that under the event $ \mathcal{N}_x $: 
\be \label{313}
\left|  \Sigma(x;E)  - \sigma(x;E)  \right| 
\  \leq \     \left| \tau(0,x;E)\right| \, .
\ee  
The lower bound \eqref{eq:r_1_2} follows by combining \eqref{313}   with \eqref{eq:defratioresloc} and the conditions defining $ \mathcal{E}_x $ and $\mathcal{T}_x $.   \\   

Considering the  possible effects of resonant delocalization we  obtain  the following result concerning conditions  inducing continuous spectrum in specified energy regimes.  Essential role in the proof is plaid by  the  Simon-Wolff criterion of Theorem~\ref{thm:SimWol} boosted by the zero-one law of Theorem~\ref{thm:0-1}.   The assumed regularity assumption will be expressed invoking the convolution: 
\be 
 \quad (1_\varepsilon * \varrho)(v) :=  \frac{1}{2\varepsilon} \int_{-\varepsilon}^\varepsilon \varrho(v+w) \, dw  \, . 
\ee

\begin{theorem}   \label{thm:resdeloc}
Let  $ \G $ be an infinite transitive graph and $ H(\omega) =A + V(\omega) $  a homogeneous random   operator on $ \ell^2(\G) $
 with 
$V$ an iid potential whose single-site distribution is 
absolutely continuous with density satisfying:
\be\label{density_assump}
\varrho \leq c \, \inf_{\varepsilon \in (0,\delta)} (1_\varepsilon * \varrho) \,  
\ee
at some $\delta >0$ and $c <\infty$. 
 Then  for any Borel set  $I \subset \R$ and   vertex    $ 0 \in \G $  a sufficient condition for 
\be \label{eq:resDeloc}
 \qquad \mu_{\delta_{0}}^{(c)}(I; \omega)   \ >\  0  \,  
\ee
is that  the following conditions 
 {\bf A1-3}  hold for a positive Lebesgue measure subset of  values $ E  \in I  $.
\begin{enumerate} 
\item[{\bf A1}] The functions $T$ and $t$, of which the first is defined by \eqref{t_T} and $t(u,v;E)$ is taken to depend only on $d(u,v)$, satisfy :
\be  \label{S_diverges}
  \lim_{d(x,0) \to \infty}   T(x,0;E)    = 0 \, \quad \mbox{and} \quad 
 \lim_{R\to \infty}  \sum_{x: d(0,x) = R  } t(0,x;E)\   =\  \infty \, . 
 \ee

\item[{\bf A2}] 
At all large enough $ R $, for all $x$ with  $d(0,x) = R$:
  \be  \label{T_cond2}
\sum_{y: \, d(0,y) = R  } T(x,y;E) \  \leq \   C_T   \sum_{x: \, d(0,x) = R  } t(0,x;E)  
    \,  \quad \mbox{ with some $C_T < \infty$}. 
\ee
\item[{\bf A3}] The operator has a well defined and strictly positive  density of states at energy $E$:
\be  \label{n_E_def}
n(E) 
  \ = \  
\lim_{\eta \downarrow 0 } \frac{1}{\pi}  
  \mathbb{E}\left[ \Im \langle \delta_x,  \frac{1}{H - E -i \eta} \, \delta_x \rangle \right] 
  \in (0, \infty ) \,. 
  \ee
\end{enumerate}
\end{theorem}

Before we turn to the proof let us comment on the feasibility of the assumed
conditions {\bf A1-2}, as it is illustrated on  
the case where $\G$ is a regular tree graph   
of constant degree $K+1$ and $ A $ is the graph's adjacency operator.   
 
On tree graphs the tunneling amplitude $\tau(x,y;E)$  factorizes into 
a  product of similar but shifted random variables, which can be associated with 
the decomposition of arbitrary paths from $x$ to $y$ into a sequence of `no-return' steps.   
One then finds \cite[Thm.~3.2]{AiWa_resdeloc} that both $t(x,y;E)$ and $T(x,y;E)$ decay exponentially: 
\begin{eqnarray} 
t(x,y;E) \ \leq \  C_0 \, e^{- \mathcal{L}_0 (E)\, \dist (x,y)} \notag \\ 
\\[-2ex] 
T(x,y;E) \ \leq \  C_1 \, e^{-  \mathcal{L}_1 (E) \dist (x,y) }\,  . \notag 
\end{eqnarray}  
with $  \mathcal{L}_0(E) $ which can be identified as the Lyapunov exponent of a transfer-matrix driven dynamics, and $  \mathcal{L}_1 (E) \ \leq  \   \mathcal{L}_0 (E)  $ by \eqref{t_T}.
In that situation condition {\bf A1} requires 
\be \label{Lyap<logK}
\mathcal L _0(E) \ < \ \log K  \, ,
\ee 
which is satisfied when the rate of typical tunneling decay is below the rate of  the (geometric) surface growth.   
On regular tree graphs  also assumption~{\bf A2} is valid when  \eqref{Lyap<logK} holds.  This follows from the inequality~\cite[Thm.~3.2]{AiWa_resdeloc}
\be \label{Lyap_1}
   \mathcal L _1(E) \geq  \log \sqrt{K}  
   \ee
and the hyperbolic geometry of the tree: for each $x$ with $\dist(x,0)=R$ most of the $R$ sphere is asymptotically   further from $x$ than from the center $0$ by a factor which asymptotically can be chosen arbitrarily close to $2$.   Due to this,  the surface average (over $y$) of $T(x,y;E)$ decays at asymptotically twice the decay rate of the surface average of  $T(0,y;E)$ (the exact calculation is elementary).\\

The zero-disorder values of the 
the above pair of   exponents are, for regular tree graphs:
\be 
\mathcal{L}_0(E) \ = \ \mathcal{L}_1(E) \ = \ 
\begin{cases}
       \log \sqrt{K}  & \mbox{for $ E\in  [-2\sqrt K, 2 \sqrt K]  = \sigma(A) $}   \\[2ex]    
              \log K   & \mbox{for $ |E| = [-(K+1), K+1] $} 
              \end{cases}
\ee 
These  expresses the facts that: i) there are no fluctuations, ii) over the spectrum of $A$ the $\ell^2$-sum 
$\sum_{x:\, \dist(x,0)=R} |G_0(0,x;E)|^2$ tends to a finite constant  as $R\to \infty$ iii)  the Lyaponov exponent grows as the distance of $E$ from the spectrum increases. 
Thus, in this case 
sufficient  Lyaponov exponent bounds are in place, and partial continuity arguments for the exponents at positive disorder can be developed, making Theorem~\ref{thm:resdeloc} applicable  at weak disorder throughout --and even beyond-- the spectrum of $A$~\cite{AiWa_resdeloc}.\\

To establish Theorem~\ref{thm:resdeloc} it suffices to prove the following estimate for  
\be
N_R(E) := \sum_{x:\,   d(x,0) = R } \indfct_{ \mathcal{T}_x \cap  \mathcal{E}_x \cap \mathcal{N}_x} \,   
\ee
which counts the number of sites $x\in \G$ at which in a given realization there is a strong enough resonance at energy $E$   to be noted at $0$,  in the sense defined by the  conditions stated in \eqref{TEN}.

\begin{lemma}\label{lem:condrargresd}
Let   $H(\omega)$ be a random operator satisfying the assumptions of Theorem~\ref{thm:resdeloc}, and   $I\subset \R$   a Borel set such that for almost every $ E \in  I $   the conditions  
{\bf A1-A3} are satisfied and also
 \begin{equation}\label{eq:neq0}
  \mathbb{P}\left( \Im \Sigma(0;E) \  {=} \  0 \right) \ = \ 1  \,  .  
 \end{equation}
Then for any $M<\infty$, and almost every $E\in I$: 
\be\label{eq:aimforN}
\liminf_{R\to \infty}  \,  \mathbb{P}\left( N_R(E)  \geq M \right) \  \geq \ p_0(E) \, , 
 \ee
 with some  $p_0(E) >0$ which does not depend on $M$.  
\end{lemma}
%

%

Once this is proven,   Theorem~\ref{thm:resdeloc} readily follows:  
it is  already known that  the $\ell^2$-sum in \eqref{ell2_div}  diverges 
at energies at which    $  \Im \Sigma(0;E) \  {\neq} \  0 $,  
 and the Lemma allows to conclude  positive probability of the sum's divergence (under the assumptions {\bf A 1-3}) regardless of \eqref{eq:neq0}.   
The zero-one law of Theorem~\ref{thm:0-1} allows then to raise the probability in \eqref{ell2_div}  to one, and the rest follows by the    
Simon-Wolff criterion (Theorem~\ref{thm:SimWol}).\\

Let us then turn to the proof of Lemma~\ref{lem:condrargresd}.  

\subsection{The second-moment proof}   

Lemma~\ref{lem:condrargresd} is proved below through a two step argument: the first is to show that   the mean (i.e. first moment of $N_R$)  diverges (as expressed in \eqref{eq:avNdiverges}).   
Then, the second-moment test will be used to establish a uniformly positive lower bound on the probability the random variables $N_R$  assume values comparable with their mean.  The alternative which  needs to  be ruled out here is that the mean diverges only due to very large contributions of very rare events, while the typical  range of values (e.g. the median) remains finite.   
A convenient tool for such purpose  is the Paley-Zygmund inequality, which states that
\begin{equation}\label{eq:PZ}
\mathbb{P}\left( \, N\geq \theta \, \mathbb{E}\left[N\right] \, \right) \geq (1-\theta)^2 \, \frac{ \mathbb{E}\left[N\right]^2}{ \mathbb{E}[N^2]} \, . 
\end{equation}
for any random variable~$ N $ and any $ \theta \in (0,1) $. 
To  employ it, one needs to derive  a lower bound on $ \mathbb{E}\left[N_R\right] $ and an upper bound on $ \mathbb{E}[N_R^2]$.

\subsubsection*{The lower bound} 

The first of the two steps outlined above is: 
 
\begin{lemma}\label{lem:lowerbd} 
Under the assumptions of Theorem~\ref{thm:resdeloc}  for Lebesgue-almost all $ E \in  \R $ at which \eqref{eq:neq0}  and  {\bf A1-3} hold: 
\be\label{eq:avNdiverges} 
\lim_{R\to \infty} \mathbb{E}\left[    N_R\right] = \infty \, .
\ee 
\end{lemma}
\begin{proof}
Let us first consider the probability of the rare event  $ \mathbb{P}\left(\mathcal{E}_x \right)$.
Applying  the shift covariance, one gets   
\begin{eqnarray}  \label{E_ZP} 
\liminf_{ d(x,0)\to \infty   } \  \frac{  \mathbb{P}\left(   \mathcal{E}_x   \right) }{2 t(0,x;E)} &  =  &     
  \liminf_{\varepsilon \downarrow 0 } \frac{1}{2\varepsilon}  
  \mathbb{E}\left[ \int \indfct[|v - \Sigma(x;E)|\leq \varepsilon ] \, \varrho(v) dv\right]   \, . 
\end{eqnarray}
On the other hand, by the rank one perturbation formula~\eqref{rank_1}, the density of states function which is given by \eqref{n_E_def} can be presented as
\be \label{E_Za}
n(E)   \ = \ 
   \lim_{\eta \downarrow 0 } \frac{1}{\pi}  
  \mathbb{E}\left[ \Im \frac{1}{V(x) - \Sigma(x;E+i \eta) }\right] \,. 
\ee
Since 
$\lim_{\eta \to 0} \Sigma(x;E+i \eta) = \Sigma(x;E) \in \R $  (in the distributional sense), the delta function principle which is stated here more explicitly in Lemma~\ref{lem:rho_f} is applicable, and it implies a relation between the  last expressions in  \eqref{E_ZP} and \eqref{E_Za}, with the consequence that  
\be \label{E_Z} 
\liminf_{  d(x,0) \to \infty   } \  \frac{  \mathbb{P}\left(   \mathcal{E}_x   \right) }{2 t(0,x;E)} \ \geq   \  n(E) \, .  
\ee

\medskip 

Let now $\{  \mathcal{Z}_x \}_{ x \in \G} $ be family of events for which:
\begin{enumerate}[1.]
\item  $ \mathcal{Z}_x $ is independent of $ V(x) $ for all $ x \in \G $, and
\item $ \displaystyle \lim_{R \to \infty} \max_{x: d(x,0)=R}   \mathbb{P}\left( \mathcal{Z}_x^c  \right) = 0 $, 
\end{enumerate} 
and  $ E \in  \R $ an energy at which \eqref{eq:neq0}  and  {\bf A1-3} hold.   
One may evaluate the joint probability by first conditioning on $ V_{\neq x } $:
\begin{align}  
\mathbb{P}\left(  \mathcal{E}_x \cap  \mathcal{Z}_x^c\right) \ & = \mathbb{E}\left[ \indfct_{ \mathcal{Z}_x^c} \mathbb{P}\left( \mathcal{E}_x \, | \, V_{\neq x} \right) \right]
= \mathbb{E}\left[ \indfct_{ \mathcal{Z}_x^c} \int \indfct[|v - \Sigma(x;E)|\leq t(0,x;E)] \, \varrho(v) dv\right] \notag \\
& \leq \ 2 t(0,x;E) \, \|\varrho \|_\infty \, \mathbb{P}\left( \mathcal{Z}_x^c\right) \, . \notag
\end{align}
This implies $ \lim_{ d(x,0) \to \infty   }  \mathbb{P}\left(  \mathcal{E}_x \cap  \mathcal{Z}_x^c\right)/  t(0,x;E) = 0 $. Since $ \mathbb{P}\left(   \mathcal{E}_x \cap  \mathcal{Z}_x  \right) =  \mathbb{P}\left(   \mathcal{E}_x   \right) -  \mathbb{P}\left(   \mathcal{E}_x \cap  \mathcal{Z}_x^c  \right)  $, we thus get the following extension of \eqref{E_Z}:
\be \label{E_Z_b} 
\liminf_{ d(x,0) \to \infty   } \  \frac{  \mathbb{P}\left(   \mathcal{E}_x \cap  \mathcal{Z}_x  \right) }{2 t(0,x;E)} \ =  \  
\liminf_{ d(x,0) \to \infty   } \  \frac{  \mathbb{P}\left(   \mathcal{E}_x   \right) }{2 t(0,x;E)} \ \geq   \  n(E) \, . 
\ee

To verify that the above  applies to $ \mathcal{Z}_x = \mathcal{T}_x \cap \mathcal{N}_x $, we note that by its definition this event is independent of $ V(x) $. To check the other assumption let us note the following:
\begin{enumerate}[1.]
\item  By our selection of the function $ t(0,x)  $,  for all $ x \in \G $:
\begin{align}\label{eq:BoundTbelow}
\lim_{R\to \infty} \min_{d(x,0)=R} \mathbb{P}\left( \mathcal{T}_x^c\right)  \  = \ 1   \, .
\end{align}
 \item Since the random variable $ V(0) $ is independent of 
 $ \sigma(0;E) $ and its distribution is  absolutely continuous with density $ \varrho \in L^\infty(\R)  $, we also have:
 \be\label{eq:BoundNbelow}
 \mathbb{P}\left(  \mathcal{N}_x ^c \right) \leq \mathbb{E}\left[ \min\{ 2 \| \varrho \|_\infty  |\tau(x,0;E)|, 1 \} \right]   \, . 
 \ee
The right side turns to zero uniformly for all $x$ with $\dist(0,x) = R$ in the limit $ R \to \infty $ by the assumption ${\bf A2}$. 
\end{enumerate}

By \eqref{E_Z} we may now conclude that for all  $ R \in  \mathbb{N} $ sufficiently large and all $ x$ with $ d(x,0) = R $,
\be  \notag 
 \mathbb{P}\left( \mathcal{T}_x \cap  \mathcal{E}_x \cap  \mathcal{N}_x  \right)  \geq \frac{n(E)}{2} \, t(0,x;E) \, ,
\ee
and hence
\be \label{N_mean}
\mathbb{E}\left[    N_R\right] \  = \ \sum_{x:  d(x,0)  = R } \mathbb{P}\left( \mathcal{T}_x \cap  \mathcal{E}_x \cap  \mathcal{N}_x  \right) \ \geq \  \frac{n(E)}{2}   \ \sum_{x:  d(x,0) = R }  t(0,x) \, .
\ee
The claimed divergence \eqref{eq:avNdiverges}, for  $R\to \infty$,  is then  implied by~\eqref{S_diverges} of  assumption {\bf A1}. 
\end{proof} 

\medskip 

\subsubsection*{The upper bound} 

For the second moment (upper) bound we start from: 
 \begin{align} 
 \notag 
\mathbb{E}\left[    N_R (N_R-1) \right] \ & = \sum_{x\neq y}^{\qquad(R)} \mathbb{P}\left( \mathcal{T}_x \cap  \mathcal{E}_x \cap  \mathcal{N}_x \cap \mathcal{T}_y \cap  \mathcal{E}_y \cap  \mathcal{N}_y \right) 
\leq \sum_{x\neq y}^{\qquad(R)} \mathbb{P}\left( \mathcal{E}_x \cap  \mathcal{E}_y\right) \, ,
\end{align}
with the susums in $\sum^{(R)}$ are  over sites of $\G$ at distance $R$ from $0$.   

By \eqref{eq:defSigmardel} the event $\mathcal{E}_x$ corresponds to $\left\{ \, |G(x,x;E+i0) | \ge 1/t(0,x;E)\, \right \}$, and likewise for $y$.  The challenge is to bound the effects of correlations between such rare events.
Considering the  restriction of the resolvent kernel to the two dimensional space spanned by $\delta_y$ and $\delta_y$, we have the following estimate  (which forms a slight variant of \cite[Thm.~A.2]{AiWa_resdeloc}).

\begin{lemma}\label{lem:ffm2cor}
Let $ A $ be a self-adjoint operator in $  \ell^2(\mathbb{G})  $ and  
  $\rho(du\, dv) = \varrho_1(u) \, \varrho_2(v) du\, dv $ an absolutely continuous probability measure on $\R^2$ 
with bounded densities $ \varrho_{j} \in L^\infty(\R) $ ($j=1,2$). 
Then there is some $ C < \infty $ such that the Green function of
 $$ H_{u,v} := A + u \indfct_{\{y\}} + v  \indfct_{\{x\}} $$  
  satisfies for any 
 $z \in \C\backslash \R$ and any $ a,b > 0 $:
 \begin{multline}\label{eq:twoL1}
 \rho\left(\left\{ (u,v) \in \R^2 \ | \  |G_{u,v}(x,x;z)| > a^{-1}  \; \mbox{and} \;  |G_{u,v}(y,y;z)| >  b^{-1}  \right\}\right) \\
  \leq \  4 \| \varrho\|_\infty^2 \sqrt{ ab }   \, 
 \min\Big\{ 2 ( \sqrt{ab} +  \sqrt{|\tau(x,y;z)| |\tau(y,x;z)|} ) , \max\Big\{ \sqrt{\frac{a}{b}} , \sqrt{\frac{b}{a}} \Big\}\Big\}  \, . 
\end{multline}
with $ \tau(x,y;z) $  the tunneling amplitude associated with $ (\delta_x,\delta_y) $ at $ z $ and $ \| \varrho\|_\infty := \max\{ \| \varrho_1\|_\infty, \| \varrho_2\|_\infty \} $
\end{lemma} 
\begin{proof}
The rank-$2 $ Schur complement formula \eqref{eq:rank2}  reveals the dependence of the diagonal Green functions on $(u,v) $. 
Abbreviating
 \begin{align} 
U \ := \ \sqrt{\frac{b}{a} }\left( u - \sigma(x;z)\right) \, , \quad 
V \ := \  \sqrt{\frac{a}{b} }\left( v- \sigma(y;z)\right)  \, ,
\end{align}   
and $\gamma := \tau(x,y;z) \tau(y,x;z)$, 
The lower bounds on $|G_{u,v}(x,x;z)| $ and $ |G_{u,v}(y,y;z)|$  translate to: 
 \begin{align}
\left|U - \frac{\gamma}{V}\right| \ \le \ \sqrt{ab}\, , \quad 
\left|V - \frac{\gamma}{U}\right| \ & \le \ \sqrt{ab}   \, .    \label{v}  
\end{align}

\begin{figure}[h]  
\begin{center}
\includegraphics[width=.55\textwidth] {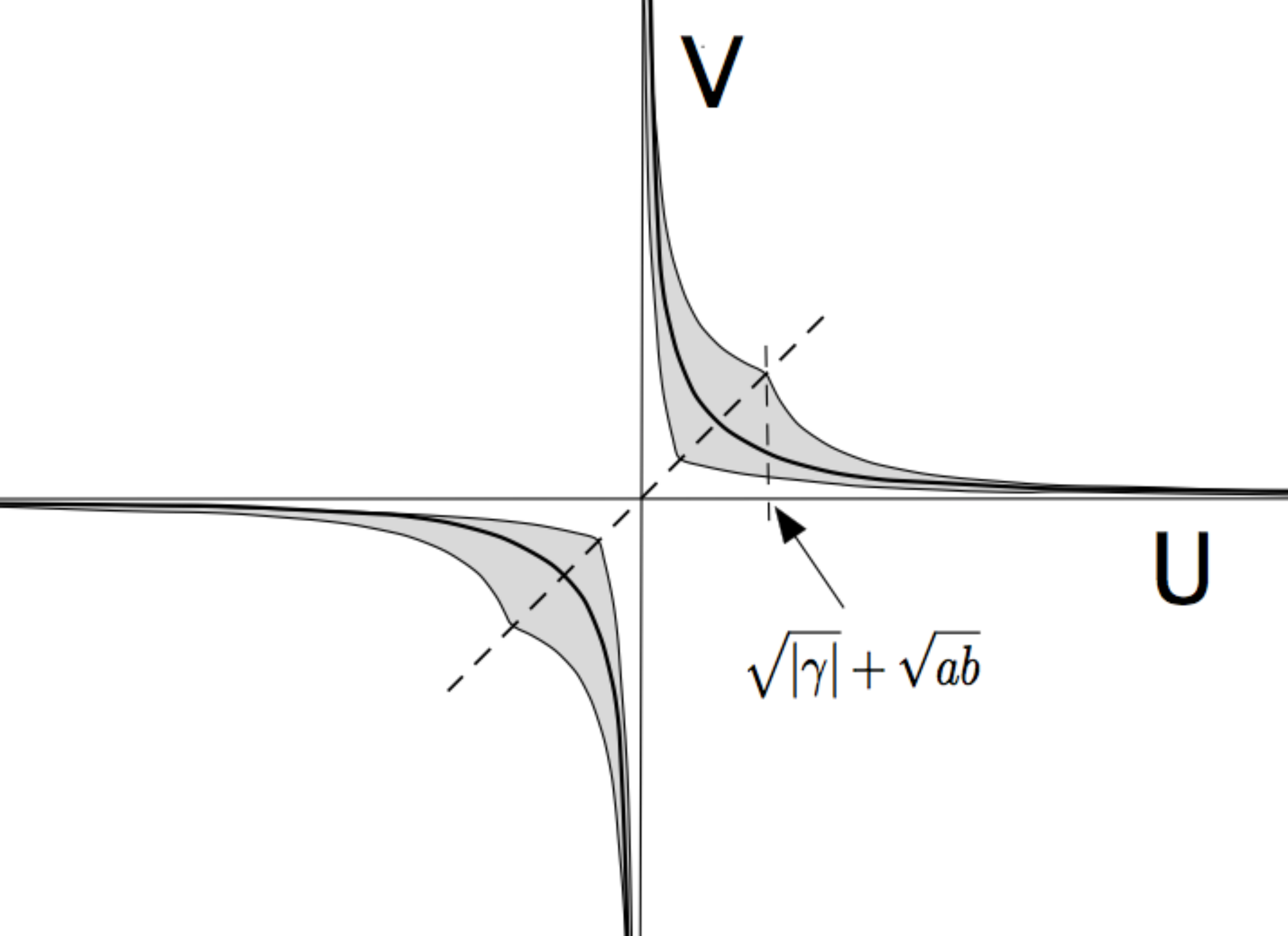} 
\caption {The solution set of \eqref{eq:twoL1} in the $(U,V)$ plane (for real $\gamma$ and $\sigma$).}    
 \label{fig_15_UV}
 \end{center}
\end{figure}

The claim can be proven through the following two observations about the set in the  $(U,V)$ plane over which the conditions \eqref{v}  are met  
(the set's shape  is indicated in Figure~\ref{fig_15_UV}).
   
\begin{enumerate}[i.] 
\item  For any solution:\\[-4ex] 
\be  \label{w-bound}
\min\{ |U|,|V|\} \ \le \ \sqrt{|\gamma|} + \sqrt{ab} \, . 
\ee 
\item For specified  $v$, 
the set of $U$  for which \eqref{v} holds is an interval of length 
at most $2  \sqrt{ab}$, 
and a similar statement holds for  $V$ and $U$ interchanged.   
\end{enumerate} 
The area bound which these yield upon integration translates directly into \eqref{eq:twoL1}.  

\end{proof}

\medskip 

We now return to the main result, which as was explained above hinges on Lemma~\ref{lem:condrargresd}.

\begin{proof}[Proof of Lemma~\ref{lem:condrargresd}] 
 Lemma~\ref{lem:ffm2cor} yields:
\be \label{exp_val}
\mathbb{P}\left( \mathcal{E}_x \cap  \mathcal{E}_y\right) \leq  8 \left(  t(0,x)  t(0,y)  +  \sqrt{ t(0,x)  t(0,y)} \; \mathbb{E}\left[\min\left\{ \sqrt{|\tau(x,y;E)\tau(y,x;E)|} , 1 \right\} \right] \right)  \, . 
\ee

The assumption \eqref{eq:neq0} allows to conclude that $\Im \langle \psi, (H-E-i0)^{-1} \psi \rangle =0$ for all $\psi \in \ell^2(\G)$ and thus:  $ |\tau(x,y;E)| = |\tau(y,x;E)| $.  
Thus the expectation value in \eqref{exp_val} coincides with $ T(x,y) $ of \eqref{t_T}.  Applying the Cauchy-Schwarz inequality we get  
\be
\mathbb{E}\left[    N_R (N_R-1) \right]  \leq 8 \, \| \varrho\|_\infty^2\,  \mathbb{E}\left[    N_R\right]^2  + 8 \, \| \varrho\|_\infty^2 \sum_{x: \floor{ d(x,0)}= R} t(0,x) \sum_{y :\floor{ d(y,0) } = R}  T(x,y)\, . 
\ee
Assumption {\bf A2}  then allows to conclude the second-moment bound:
\be \label{2nd_moment} 
\limsup_{L\to \infty} \frac{\mathbb{E}\left[    N_R (N_R-1) \right] }{\mathbb{E}\left[    N_R\right]^2} \leq 8  \, \| \varrho\|_\infty^2 \, (1+C_T)\, . 
\ee

Through the Paley-Zygmund criterion~\eqref{eq:PZ}, 
the pair of moment bounds~\eqref{eq:avNdiverges} and \eqref{2nd_moment} yield  Lemma~\ref{lem:condrargresd}.
\end{proof}

As was noted below the statement of the just proven Lemma, it readily implies Theorem~\ref{thm:resdeloc}.

\section{Discussion: exploring  the argument's limits }

\noindent 1.~For a sense of how far can the above analysis can be extended, let us consider the following, admittedly ``optimistic'', picture of the possible  reach of the resonant delocalization argument.   \\  

For an infinite transitive metric graph, let: 
\be  \notag
\chi(R)  \ := \  \log \left[ \rm{card} \{ x\in \G :\, \dist(0,x) \in [ R, R+1] \} \right] \, 
\ee
so that by definition the number of sites at distance $R$ from $0\in \G$ grows as $e^{\chi (R)}$.    E.g., on a regular tree of degree $K+1$, $\chi(R)$ grows as $ \chi(R) \approx  R   \log K $, 
  while for the $N$-hypercube the analogous function  grows much faster, as long as $ R \le N/2$. \\ 
  
Consider now the case where the tunneling amplitude is exponentially small in $\chi(\dist(0,x))$:   
\be \notag 
\tau(0,x) \approx  \exp(- [L +o(1)] \chi(\dist(0,x)) )  
\ee
 but the effective decay rate assumes, in addition to its typical value $L_0$ also a range of values $L <  L_0$ , though at  only at a random and exponentially small fraction of sites on the $R$-sphere.   To be more specific, assume it exhibits  large deviation behavior in the sense  that   for any fixed $\delta>0$:
  \be 
 \Pr \left\{  \left| \frac{\log \tau(0,x)}{  \chi(\dist(0,x))} +L \right | < \delta  \right \} 
 \approx e^{-[\gamma(L) + o(1)] \, \chi(\dist(0,x))}
 \ee  
  with  a good rate function  $\gamma(L)$.    
(The standard vocabulary  of the large deviation theory, and its relevant basic results may be found in e.g., \cite{LargeDev}.) \\  

To estimate the effect of possible resonances between $0$ and points on the subset of the $R$-sphere to which the tunneling amplitude $\tau(o,x)$ larger than normal,  let us consider the three events  described by \eqref{TEN} with the cutoff function modified to: 
\be \label{modified_t}
t(0,x) \ = \ e^{-L \, \chi(\dist(0,x)) }  \, ,   
\ee 
at a fixed $L<L_0$.   Compared with the analysis in Section 3, $\mathcal T_x$ is now made into a  rare event, of probability $\approx e^{-\gamma(L)}$.  However, for the sites where it is realized, the probability of $\mathcal E_x$ (which may still be expected to be of the order $e^{-L \, \chi(\dist(0,x))}$)  while still small, will be much larger than the previous $e^{-L_0 \, \chi(\dist(0,x))}$.    Assuming this part of the argument could be carried through,   instead of \eqref{N_mean} one would get for 
the first-moment a lower bound of the form: 
\be 
\Ev{N_R} \ \geq \  \rm{C}_\delta   e^{-\delta R} \, \,  e^{ \chi(R) }\, e^{- [\gamma(L) + L]\,  \chi(R)} \, 
\ee 
where the two exponential factors are the surface area, and the fraction of sites at which the modified events $\mathcal T_x \cap \mathcal E_x$ occur.  \\ 

Thus, assuming the validity of the above sketched large deviation structure, and considerations,  a sufficient  condition for passing the first moment  test \eqref{eq:avNdiverges} may well be:
\be  
 - \varphi (1) := \inf_{L\geq 0} [\gamma(L)   +  L ]  \ <   \ 1  \, ,   
\ee 
with $\varphi(\cdot)$  the Legendre transform of the function $\gamma$: (interpreted as $\gamma(L)=+\infty$ for $L<0$):
\be  \notag
- \varphi(s)  \ : = \  \inf_{L\geq 0} [\gamma(L)   + s  L ] \,. 
\ee 
(The regretable negative sign is to keep notational consistency with \cite{AiWa_resdeloc}.)  \\ 

Under the above assumptions (in particular applicability of the large deviation picture), for $s<1$:   
\begin{eqnarray}  
\varphi(s) & = &  \lim_{\dist(0,x) \to \infty} \frac{1}{\chi(\dist(0,x))} \log \Ev{|\tau(0,x;E)|^s}  \notag \\[2ex] 
& =&   \  \lim_{\dist(0,x) \to \infty} \frac{1}{\chi(\dist(0,x))} \log \Ev{\min[|\tau(0,x;E)|,1]^s}    \, .  \notag 
\end{eqnarray}  
The restriction to $s<1$ is to avoid the spurious effects of the trivial divergence of the first moment.  The truncation makes no difference for $s<1$,  since in that case the  moments are  bounded uniformly in $x$, however it allows to extend the definition to  $s=1$. \\ 

Thus, the first-moment test  for \emph{delocalization} could conceivably be proven, under the above assumptions, 
for the regime in which
\be  \label{eq:ell1_deloc}
\quad   \sum_x  e^{\varepsilon\,  d(o,x)} \, 
\Ev{ \,  \min \{ \, |\tau(0,x;E)| \, , \, 1\}  \, } \  = \ \infty \, 
  \,.  
\ee 
for  some $\varepsilon >0 $.   

However, to turn the above discussion into a proof one would need to establish also   the second-moment upper bound \eqref{eq:PZ}.  That would be more involved than what was faced in Section 3, since it now requires to bound correlations between not just pairs of essentially local events, at $(x,y)$ as above, but also between large deviations of the corresponding path-related tunneling amplitudes.  For tree graphs such a program was carried out in \cite{AiWa_resdeloc}.\\  

 It is of interest to note that   \eqref{eq:ell1_deloc} 
has the appearance of a  \emph{complementary} condition (except for transitional points)  to the  fractional moment \emph{localization criterion}, by which pure point spectrum in a Borel set   of positive measure $I\subset \R$  can be concluded from the property:
\be  \label{eq:ell1_loc}
 \mbox{for some $s<1 $, and all $E\in I$}: \quad   \sum_x  
\Ev{ \,  |\tau(0,x;E)   |^{s}  \, }\  < \ \infty \, .
\ee 
 (Missing is a proof that the transition between convergence and divergence occurs along a reasonably regular boundary in the $(E,\lambda)$ plane).  \\ 
 
\noindent 
2.~Although the notion of continuous spectrum applies only to spectra of operators on infinite graphs, resonant delocalization can play an essential role also for  finite graphs.  In particular, that was shown to be the case for the complete graph on $1\ll N<\infty$ points \cite{ASW_14}. 
It would be of interest to see the present techniques explored further in the finite graph context.  \\

        

\appendix

\section*{APPENDIX}

\section {Stochastic delta function principle} 

Following is the delta function principle which is invoked  in the proof of Lemma~\ref{lem:lowerbd}.  
\begin{lemma}\label{lem:rho_f} 
Let  $\Xi_n = X_n+i Y_n$ a sequence of random variables with values in $\C^+$  which converge in distribution to a real-valued random variable $ X $ and let $V$ be an independent  random variable of an absolutely continuous distribution which for some $\delta >0$ and $c(\delta) <\infty$ satisfies the pointwise bound:
\be\label{eq:assdensity}
\rho \  \leq \  c(\delta) \, \inf_{\varepsilon \in (0,\delta)} (1_\varepsilon * \rho) \, .
\ee
Then:
\be \label{goal}
\lim_{n\to \infty} \,   \mathbb{E}\left[\Im \frac{\pi^{-1}}{V-\Xi_n }    \right]  \ \leq   \  c(\delta)  \inf_{\varepsilon \in [0,\delta]} \frac{1}{2\varepsilon} \,   \mathbb{P}\left(  | V- X|\leq \varepsilon\right) \, .
 \ee 
\end{lemma}
\begin{proof}
We will make use of the fact that for any $ \rho \in L^1 \cap L^\infty $ the family of functions $  f_\varepsilon : \R \times [0,\infty) \to \R $ defined though
\be   \notag 
 f_\varepsilon(\alpha,\beta)  :=   \begin{cases} \int (1_\varepsilon * \rho)(v)  \frac{\beta}{(v-\alpha)^2 +\beta^2} \, \frac{dv}{\pi} \, , & \quad \beta > 0 \, , \\
\;  (1_\varepsilon * \rho)(\alpha) \, , & \quad \beta = 0 \, .
 \end{cases}
\ee
is bounded and continuous for any $ \varepsilon > 0 $. Its boundedness is evident. For a proof of continuity, we use its Fourier representation 
\be\label{eq:Fourier}
 f_\varepsilon(\alpha,\beta) =  \int \hat \rho(k) \,  {\rm sinc}(k\varepsilon) \, e^{i\alpha k - \beta |k|}  \, \frac{dk}{2\pi} \, , 
\ee
where ${\rm sinc}(\xi) := \frac{\sin(\xi)}{\xi} $ if $ \xi \neq 0 $ and $ {\rm sinc}(0) :=  1 $. The function $  \hat \rho(k) := \int \rho(v) e^{ikv} \frac{dv}{2\pi}  $  is  square-integrable  since $ \rho \in L^1 \cap L^\infty $ (the latter by \eqref{eq:assdensity}). This renders the integrand in~\eqref{eq:Fourier} absolutely integrable even in case $ \beta = 0 $. The claimed continuity  follows then from~\eqref{eq:Fourier}, with the help of the dominated convergence theorem. 

By assumption, the sequence of probability measures on $ \R \times [0,\infty) $ 
\be \notag 
 \mu_{n}(d\alpha d\beta) = \mathbb{P}\left( \Xi_n \in d\alpha + i d\beta \right) 
 \ee  converges weakly to   $ \mu(d\alpha ) \, \delta_{0}(d\beta) $, with $ \mu $ the probability distribution of $ X $ and $ \delta_0 $ Dirac's measure at zero. Using the estimate
\begin{align} \notag 
\mathbb{E}\left[\Im \frac{\pi^{-1}}{V-\Xi_n }    \right]  = \int \int \rho(v)  \frac{\pi^{-1} \beta}{(v-\alpha)^2 +\beta^2} dv \, \mu_{n}(d\alpha d\beta) \leq 
c(\delta)  \int  f_\varepsilon(\alpha,\beta) \, \mu_{\eta}(d\alpha d\beta) \, ,
\end{align}
the claim follows, since by the convergence of $\mu_n$ and the continuity of $ f_\varepsilon $:
 \be  \notag 
 \lim_{n\to \infty}  \int  f_\varepsilon(\alpha,\beta) \, \mu_{n}(d\alpha d\beta) = \int f_\varepsilon(\alpha,0) \mu(d\alpha ) = \frac{1}{2\varepsilon} \mathbb{P}( | V- X | \leq \varepsilon) \,   
 \ee
 for any $\varepsilon \in (0, \delta)$.  
\end{proof}

\section{Simplicity of the pure point spectrum}  

The  rank-one  (and more generally finite rank) perturbation method which underlines the above criteria for continuous spectrum   allows also to shed light  on properties of the pure point spectrum and in particular its almost sure simplicity.   
Results in this vane were initially presented by 
B. Simon, and then  V. Jak\v{s}i\'c and Y. Last, who proved simplicity of the point spectrum~\cite{S} and more generally singular spectrum~\cite{JLdm}  for operators with random potential of absolutely continuous conditional distribution.   
The following streamlined statement shows that  
the absolute continuity of the distribution of $V(x)$ (which enables one to apply the spectral averaging principle)  
is an unnecessarily strong condition  for the simplicity of the point spectrum.

\begin{theorem}\label{thm:simplicity}
Let $ H(\omega) = A + V(\omega) $ be a random operator  on $ \ell^2(\G) $ with $ A $ bounded and self-adjoint  and $ V(\omega) $ a random potential such that for any  $ x   \in \G $   the conditional distribution of $V(x)$ conditioned on $ V_{\neq x} := \{ V(u) \}_{u \neq x} $ is continuous. 
Then  the pure point spectrum of $H(\omega) $ is almost surely simple.
\end{theorem}

Our proof  proceeds through Proposition~\ref{prop:Aronj}, of Simon-Wolff \cite{SimWol},      
extended by the observation that under its condition   the vector 
$  (H_0 - E -i 0 ) ^{-1} \, \psi  $  provides  the unique (up to a multiplicative constant)  proper eigenfunction of $H_v= H_0  + v \, P_\psi$ within its cyclic subspace  
\be \notag 
\mathcal H_\psi := \overline{\rm{span} }\{ (H_v -\zeta)^{-1} \psi \, :\, z\in \C\backslash \R \} \subset  \ell^2(\G) \, .
\ee

\begin{lemma}\label{lem:specnullav}
Let $\G$ be a countable set,  $H_0$ a bounded self-adjoint operator in $\ell^2(\G)$, and $H_v$ the one-parameter family of operators defined by \eqref{eq:H_v} with $\psi\in \ell^2(\G)$.    
 Then for any countable subset $ S \subset \R $ and any probability measure $ \rho $ which is continuous 
 \begin{equation}\label{eq:specnullav}
 \int \mu_{v,\psi}(S) \, \rho(dv) = 0 \, .
 \end{equation}
\end{lemma}
\begin{proof}
By the countable additivity of the measure $ \int \mu_{v,\psi}(\cdot) \, \rho(dv)$ it suffices to prove \eqref{eq:specnullav}  for one-point sets $ S = \{ E \} $, at  arbitrary  $ E \in \R $.   
The contribution to the integral from the one point set  $v=0$ vanishes since $ \rho(\{0\}) = 0 $.    For $v\neq0$,  by Proposition~\ref{prop:Aronj}  the integrand does not vanish only if 
$v = - \langle \psi , \, (H_0 - E -i 0 ) ^{-1} \, \psi \rangle^{-1} $. 
However, this point
is also not charged, since $\rho$ is a continuous measure.  Hence the integral vanishes. 
\end{proof}

\begin{proof}[Proof of Theorem~\ref{thm:simplicity}] 
For  $x\in \G$ let 
\begin{multline}  
\Omega_x := \\ 
\left\{ \omega \in \Omega \, | \,  \, \mbox{for some $ E \in \R $:  
$ {\rm dim} \ {\rm range}\,  P_{\{E\}}(H(\omega))  \geq 2 $ and 
$P_{\{ E\}}(H(\omega)) \delta_x \neq 0$ } \right\}
\end{multline} 
Our goal is to prove that this set is of vanishing probability.  
To highlight the dependence of the random operator 
$H=H(\omega)$ on $V(x)$ we write it in the form  
\begin{equation} \label{42}
H =: H_0 + V(x) \, \indfct_{\{x\}} \, ,  
\end{equation}
where $ H_0 $ is independent of $ V(x) $ and defined through the above equation. 
The full Hilbert space 
may be presented as the direct sum of the cyclic subspace 
$ \mathcal{H}_{H,x} = \mathcal{H}_{H_0,x} $ and its orthogonal complement:
\begin{equation}  \notag 
\ell^2(\G) = \mathcal{H}_{H_0,x} \oplus \mathcal{H}_{H_0,x}^\perp \, . 
\end{equation}
If the vector $ \delta_x $ is cyclic for $ H $ then $ \mathcal{H}_{H_0,x}^\perp $ consists of just the zero vector and the pure point spectrum of $ H $ is simple. In the following we therefore concentrate on the case that 
$ \mathcal{H}_{H_0,x}^\perp  \neq \{0\} $. 

The operator $ H $ leaves both $\mathcal{H}_{H_0,x}$ and its  orthogonal complement invariant. Its point spectrum is therefore the union of point spectra the operator has in the two subspaces. 
Two notable features of this decomposition are: i) the spectrum of $H$ in $\mathcal{H}_{H_0,x}$ is non-degenerate, ii) the spectrum in  $\mathcal{H}_{H_0,x}^\perp$ does not vary with $V(x)$. 
  
Let $S_x$ denote the set of eigenvalues of the restriction of $ H $ to $ \mathcal{H}_{H_0,x} ^\perp$.  Its independence of  $V(x)$  follows from the observation  that the eigenfunctions    $ \psi_E $ in $  \mathcal{H}_{H_0,x}^\perp  $ have to vanish at $x$,  since 
for any $ \varphi = f(H) \delta_x \in  \mathcal{H}_{H_0,x} $ 
\begin{equation}
0 = \langle \varphi \, , \, \psi_E \rangle = \langle \delta_x \, , \, f(H) \psi_E \rangle = f(E) \, \psi_E(x) \, . 
\end{equation}
This implies that $ \psi_E $ is also an eigenfunction of $ H_0 + \widehat V(x) \, \indfct_{\{x\}} $ for any other $ \widehat V(x) \in \R $ with the same eigenvalue $ E $. 

Since the set $ S_x $ is independent of $ V(x) $, by Lemma~\ref{lem:specnullav}  the conditional expectation of $\mu_{\delta_x}(S_x) $,  conditioned on $ V_{\neq x } $, is zero for each value of 
all the other parameters.  (In this argument $\mu_{\delta_x} $ is the spectral measure associated with $H$ and the vector  $ \delta_x\in \ell^2(\G) $, and use is made of the continuity of the conditional distribution of $V(x)$).   
Hence 
\begin{equation} 
\notag 
\mathbb{E}\left[   \mu_{\delta_x}(S_x) \right] = \mathbb{E}\left[\mathbb{E}\left[  \mu_{\delta_x}(S_x) \, \big| \,   V_{\neq x } \right] \right] = 0 \, . 
\end{equation}
This means that  the point spectrum of $H$ in $\mathcal{H}_{H_0,x}$,  which  supports $ \mu_{\delta_x} $, almost surely does not intersect $ S_x $, or 
\be 
\text{Prob} (\Omega_{x}) \ = \ 0 \,. 
\ee 
Since countable unions of null sets carry zero probability, also  $\Omega_0 := \bigcup_{x \in \G  }  \Omega_{x}$  has this property. 

On the complement of $\Omega_0$  for each $ E \in \R $ 
the vectors  $P_{\{E\}}(H(\omega))  \delta_x$ and $P_{\{E\}}(H(\omega))  \delta_y$ at different $x,y\in \G$ are collinear, when non-zero.   Since their collection spans the full range of 
$P_{\{E\}}(H(\omega))  $ in $\ell^2(\G)$, the point spectrum is simple for any $\omega$ in the complement of the null set $\Omega_0$. 
\end{proof}

In the simple case that $H_0\equiv 0$ and the potential is given by iid random variables, the continuity  of the distribution of $V$ is trivially both sufficient and necessary for the almost sure simplicity of the spectrum.  Hence the sufficient condition of 
Theorem~\ref{thm:simplicity}  cannot be weakened for a statement which is valid regardless of $H_0$.  
  
One may also note that Theorem~\ref{thm:simplicity} may be extended to  random potentials which behave as assumed there only along a subset $\G' \subset \G$  provided 
\be 
\overline{\rm{span} }\left \{ (H(\omega)-z)^{-1} \delta_x \, | \, z\in \C\backslash \R ,\, x\in \G' \right \} \ =\   \ell^2(\G)\,, 
\ee
 i.e., the collection of vectors $(\delta_x)_{x\in\G'} $ is cyclic under the action of  $H(\omega)$. 

More extensive   discussions of the behavior of spectra under independent rank-one perturbations can be found in~\cite{SimR1, AiWa_RO}.

\medskip

\subsection*{Acknowledgment}
This work was supported in parts by the NSF grant PHY--1305472 (MA), and the DFG grant WA 1699/2-1 (SW).  
We thank CPAM managing editor Paul Monsour for his very helpful copyediting.

\end{document}